\documentclass[pra,aps,twocolumn,superscriptaddress,longbibliography, nofootinbib, showpacs]{revtex4-1}

\usepackage[colorlinks=true, citecolor=red, urlcolor=blue ]{hyperref}
\usepackage{graphicx}
\usepackage{bm}
\usepackage{amsmath,amsfonts}
\usepackage{amsthm}
\usepackage{hyperref}
\usepackage{xcolor}
\colorlet{linkequation}{blue}
\newcommand*{\SavedEqref}{}
\let\SavedEqref\eqref
\renewcommand*{\eqref}[1]{%
  \begingroup
    \hypersetup{
      linkcolor=linkequation,
      linkbordercolor=linkequation,
    }%
    \SavedEqref{#1}%
  \endgroup
}
\hypersetup{
    urlcolor=magenta,
    citecolor=blue,
    linkcolor=blue,
    linkbordercolor=linkequation
           }

\usepackage{braket}
\usepackage{float}
\theoremstyle{plain}
\usepackage{color}
\usepackage{amssymb}
\usepackage{amsthm}
\usepackage{amsfonts}
\usepackage{float}
\usepackage{tabularx}
\usepackage{graphicx}

\usepackage{mathtools}
\usepackage{esvect}
\usepackage{wrapfig}
\usepackage{amsthm}
\usepackage{verbatim}
\usepackage{bbm}
\usepackage[normalem]{ulem}

\usepackage{enumitem}
\usepackage{fmtcount}
\usepackage{booktabs}
\usepackage{csquotes}
\usepackage{epsfig}

\usepackage{tabularx}
\usepackage{graphicx}
\usepackage{color,soul}
\usepackage{amsmath}
\usepackage{braket}
\usepackage{latexsym}
\usepackage{bm}
\usepackage{graphics,epstopdf}
\usepackage{enumitem}
\usepackage{fmtcount}
\usepackage{booktabs}
\usepackage{csquotes}
\usepackage{epsfig}

\theoremstyle{plain}

\def\bea{\begin{eqnarray}}
\def\eea{\end{eqnarray}}
\def\ba{\begin{array}}
\def\ea{\end{array}}

\def\beq{\begin{equation}}
\def\eeq{\end{equation}}

\usepackage[normalem]{ulem}
\usepackage{float}
\usepackage{graphicx} 
\usepackage{lipsum}
\usepackage{dcolumn}          
\usepackage{amssymb}
\usepackage{appendix}
\usepackage{physics}   
\usepackage{mathtools}
\usepackage{esvect}
\usepackage{wrapfig}
\usepackage{amsthm}

\usepackage{verbatim}
\usepackage{bbm}

\usepackage[mathscr]{euscript}
\def\Tr{\operatorname{Tr}}

\def\({\left(}
\def\){\right)}
\def\[{\left[}
\def\]{\right]}

\newtheorem{corollary}{Corollary}

\newtheorem{observation}{Observation}

\begin{document}
\title{Measurement incompatibility at all remote parties do not always permit Bell nonlocality}

\author{Priya Ghosh$^1$, Chirag Srivastava$^2$, Swati Choudhary$^{1,3}$, and Ujjwal Sen}
\affiliation{Harish-Chandra Research Institute,  A CI of Homi Bhabha National Institute, Chhatnag Road, Jhunsi, Prayagraj  211019, India\\
$^2$Laboratoire d'Information Quantique, CP 225, Université libre de Bruxelles (ULB),
Av. F. D. Roosevelt 50, 1050 Bruxelles, Belgium\\
$^3$Center for Quantum Science and Technology (CQST) and  
Center for Computational Natural Sciences and Bioinformatics (CCNSB),
International Institute of Information Technology Hyderabad, Prof. CR Rao Road, Gachibowli, Hyderabad 500 032, Telangana, India}

\begin{abstract}
Two important ingredients necessary for obtaining Bell nonlocal correlations between two spatially separated parties are an entangled state shared between them and an incompatible set of measurements employed by each of them. We focus on the relation of Bell nonlocality with incompatibility of the set of measurements employed by both the parties, in the two-input and two-output scenario. We first observe that Bell nonlocality can always be established when both parties employ any set of incompatible projective measurements.
On the other hand, going beyond projective measurements, we present a class of incompatible positive operator-valued measures, employed by both the observers, which can never activate Bell nonlocality. 
Furthermore, we find a sufficient criterion for achieving Bell nonlocal correlations given a fixed amount of pure two-qubit entanglement and a fixed amount of incompatibility of projective measurements applied by either both parties or a single party.
\end{abstract}

\maketitle

\section{Introduction}
\label{sec1}
Violation of Bell inequality~\cite{bell_aspect_2004, bell-review} is a phenomenon of utmost fundamental and practical importance in quantum information.
The Bell scenario investigates into whether correlations established from local measurements by spatially separated entities, can be accounted for by using local hidden variable (LHV) models. Instances of non-existence of any LHV model for explaining these correlations are termed as Bell nonlocal. 
Quantum theory predicts the existence of such Bell nonlocal correlations. For experimental validation, see Refs~\cite{bell-expt-1,bell-expt-2,bell-expt-3,bell-expt-4,bell-expt-6,bell-expt-7,bell-expt-8,bell-expt-9}.
This intriguing characteristic has practical implications, serving as a valuable resource in various quantum information applications, including device-independent quantum key distribution~\cite{bell-cryptography,scarani-qkd,pironio-qkd-asymmetric,acin-qkd,colbeck-qkd,farkas-qkd} and device-independent random number generation~\cite{Pironio2010,kent-randomness,pan-randomness,colbeck-randomness}.

Two important resources of quantum theory responsible for obtaining non-classical effects are entanglement~\cite{horodecki09,Guhne09,sreetama-titas} and incompatible measurements~\cite{incomp-busch,incomp-lahti,incomp-stano,incomp-pvm,Heinosaari10,incomp-ziman,incompatibility-review}. 
Both of them are also necessary for establishing Bell nonlocal correlations among spatially separated parties. However, none of them separately guarantee Bell nonlocality. Firstly, for some entangled states, there exists a LHV model that can predict the results for any local measurement setting~\cite{bell-ent,Barrett02} (See also~\cite{popecu-tele,popescu-hidden,HHH-PLA,gisin-hidden,peres-collective} in this regard.). On the other hand, for some measurement settings applied by one of the parties that are globally incompatible, there exists a LHV model for all shared bipartite states and all measurement settings at the other party~\cite{imcomp-bell-not-1,incomp-bell-not-2}. 
However, for the Bell scenario of two spatially separated parties where each party employs two binary outcome measurements, 
the incompatibility of measurements at one party can always ensure a Bell inequality violation, i.e., if one party performs incompatible measurements, there will always exist a shared state and a pair of measurements at the another party that result in a violation of a Bell inequality~\cite{incomp-bell}.

In the current work, we study the relationship between Bell nonlocality and the measurement incompatibility associated with all the parties in a Bell scenario of two spatially separated parties. 
Specifically, we aim to determine whether it is always possible to find a shared state between the two parties that exhibits Bell nonlocality for any pair of incompatible measurements at both parties. 
It is straightforward to see from the results of ~\cite{imcomp-bell-not-1,incomp-bell-not-2} that incompatible measurements, performed by each party, do not always allow Bell nonlocality. However, the same conclusion is not at all obvious for the Bell scenario of two parties with each party performing two binary outcome measurements. Since in this scenario, which we also refer as ``two-input and two-output Bell scenario", measurement incompatibility by one of the parties can always ensure Bell nonlocality \cite{incomp-bell}.
Here, we focus on a two-input and two-output Bell scenario and use the Clauser-Horne-Shimony-Holt (CHSH) inequality~\cite{Clauser69} to investigate Bell nonlocality.
We first observe that there will always be a shared state such that any pair of incompatible projective measurements at both parties will violate the CHSH inequality. Interestingly, however, the same is not true when both parties employ incompatible positive operator-valued measurements. I.e., we find a specific set of incompatible positive operator-valued measurements at both parties that will never show a Bell nonlocality in this Bell scenario, whatever state is shared among them.
Thus, the incompatibility of measurements at both parties do not ensure Bell nonlocality in the two-input and two-output Bell scenario as well.

In the latter half of our work, 
we focus on a two-input and two-output Bell scenario 
where the parties can share only pure two-qubit states and they can perform only local projective measurements at their sides. We try to develop a relationship between a fixed amount of shared entanglement and a fixed amount of incompatibility of measurements at party/parties in this Bell scenario for following two cases. First, we try to find a sufficient criterion required to violate the CHSH inequality for a fixed amount of shared entanglement and a fixed amount of degree of incompatibility of the measurements at both the parties. 
It helps us to find that even with an arbitrarily low degree of incompatibility of measurements associated with both parties, Bell nonlocality can still be achieved if the shared entanglement exceeds a certain threshold. 
Second, given a fixed amount of shared entanglement and a fixed amount of degree of incompatibility corresponding to the measurements at one party, we aim to find a sufficient condition to violate Bell inequality for at least one pair of measurements at another party in the same Bell scenario.
In this second case, we find that only maximally entangled states can exhibit Bell nonlocality for all incompatible projective measurements at one party.

The structure of the remaining sections in this paper unfolds as follows. In Section~\ref{sec2}, a succinct overview of various concepts such as joint measurability and Bell inequality is presented. In Section~\ref{sec3}, we show that when the measurements are positive operator-valued measurements, incompatibility of the measurements at both parties do not always guarantee  Bell nonlocality, regardless of the shared quantum state by providing an example. The subsequent section~\ref{sec4} presents a sufficient condition of CHSH inequality violation for a fixed amount of shared entanglement and a fixed amount of measurement incompatibility associated with both parties' measurements. Additionally, this section explores the relationship between a fixed amount of shared entanglement and a fixed amount of measurement incompatibility corresponding to two projective measurements at one party, ensuring CHSH inequality violation with at least one pair of projective measurements at the another party. The final section, Section~\ref{sec5}, is devoted to outlining our concluding remarks.

\section{Preliminaries}
\label{sec2}
\textit{Measurement Incompatibility:} Let us initiate our discussion by delving into the realm of incompatible measurements, as expounded in Ref.~\cite{incomp-pvm,Heinosaari10,incompatibility-review}. 
A collection of Positive Operator-Valued Measures (POVMs) denoted as $\lbrace M_{a|x}: M_{a|x}\geq0, \sum_a M_{a|x}=\mathbb{I}\rbrace$, where $\mathbb{I}$ is the identity, is deemed jointly measurable or compatible if the existence of a parent POVM $\lbrace G_\lambda \rbrace$ and a set of conditional probabilities $\lbrace p(a|x,\lambda) \rbrace$ can be established, satisfying the relation
\begin{align}
M_{a|x} = \sum_{\lambda} p(a|x,\lambda) G_\lambda,
\end{align}
for all possible $a$ and $x$. In the absence of such a relationship, the assemblage of POVMs is classified as incompatible. Projective measurements are a special case of POVMs with an additional property that $M_{a|x}M_{a'|x}=M_{a|x}\delta_{aa'}$ and the effects, $M_{a|x}$, of the measurement are called as it's projectors.
The findings presented in~\cite{Heinosaari10} established a significant connection between non-commutativity of the projectors of the projective measurements and incompatibility of corresponding measurements. Specifically, the research demonstrates that in the case where only  projective measurements are performed, compatibility of these measurements equates to pairwise commutativity of all the projectors of these projective measurements. Thus, incompatible measurements, in the context of only projective measurements, imply that at least one pair of projectors of these measurement do not commute.

\textit{Bell nonlocality}: Consider an experiment of two spatially separated parties, Alice and Bob, share a state among them and perform measurements at their local subsystems to obtain joint statistics of their outcomes. Let the measurement settings employed by Alice (Bob) is denoted by $x(y)$ and the outcomes obtained are denoted by $a(b)$, then the joint probability distribution of obtaining $a$ and $b$, given the settings employed are $x$ and $y$ by Alice and Bob, respectively, is denoted by $p(a,b|x,y)$. Now, any local hidden variable (LHV) theory says that any joint probability distribution for such an experiment must admit
\begin{align}
p(a, b | x, y) = \int_\lambda d\lambda p(\lambda) p(a | x, \lambda) p(b | y, \lambda),
\end{align}
where $\lambda$ is a hidden variable which is shared between Alice and Bob before they start the experiment and $p(\lambda)$ is it's probability distribution. Now within quantum theory, the measurement by Alice (Bob), yielding outcome $a(b)$, given the setting $x(y)$, can be denoted by positive operators, $\lbrace M_{a|x} \rbrace (\lbrace N_{b|y} \rbrace)$, constituting POVMs. Also, these measurements are performed locally on a joint system which can be described by some shared quantum state, denoted by $\rho$. Then the quantum theory tells that the joint probability distributions are given by
\begin{equation}
    p(a, b | x, y)=\langle  M_{a|x} \otimes N_{b|y} \rangle_\rho,
\end{equation}
where $\langle O \rangle_\rho$ denotes the expectation value of operator $O$ with respect to $\rho$.
J. S. Bell showed that there exist correlations in the quantum theory which do not admit any LHV model~\cite{bell_aspect_2004} and such correlations are thus termed as Bell nonlocal. 
This and many following works have, thereafter, presented many inequalities which correlations admitting LHV model must satisfy and thus any violation of such inequalities can detect Bell nonlocal correlations~\cite{Clauser69,CGLMP,oh-cv,cavalcanti-cv,wolf-cv,bell-review}.
The simplest Bell scenario that exhibits Bell nonlocal correlations is one where each observer has two measurement settings, with two possible outcomes for each setting.
In such a setting, the possible measurement inputs and outcomes can be denoted as $x,y=0,1$ and $a,b=\pm1$ and we will call such a setting as two-input and two-output Bell scenario. The Bell inequality in the simplest setting by Clauser-Horne-Shimony-Holt (CHSH)~\cite{Clauser69} is given as
\begin{equation} \label{sipahi}
 \langle A_0 \otimes B_0 +  A_0\otimes B_1  + A_1 \otimes  B_0  -  A_1 \otimes B_1 \rangle \leq 2,
\end{equation}
where $A_x=M_{+1|x}-M_{-1|x}$ and $B_y=N_{+1|y}-N_{-1|y}$ are observables corresponding to the measurement settings $x$ and $y$ by Alice and Bob, respectively and the LHS is the expectation value of the Bell operator. 
This implies that no LHV theory can have a CHSH value (defined as the expectation value of the Bell operator in the CHSH inequality) 
greater than 2. Any violation of this inequality tells that the correlation present between the observers is Bell nonlocal.
The maximum CHSH value obtained in quantum mechanics is $2\sqrt{2}$ and thus there exist correlations in quantum mechanics which are Bell nonlocal and do not admit any LHV model. This maximal value is obtained when Alice and Bob share $|\phi^+\rangle (\coloneqq \frac{|00\rangle+|11\rangle}{\sqrt{2}})$, and Alice measures in the basis of observables $\sigma_z, ~\sigma_x$, and Bob opting for the measurement in the basis of $\frac{\sigma_z + \sigma_x}{\sqrt{2}},~ \frac{\sigma_z - \sigma_x}{\sqrt{2}}$. Here, $\sigma_x$,$\sigma_y$, and $\sigma_z$ are standard Pauli operators with $|0\rangle$ and $|1\rangle$ being the eigenstates of $\sigma_z$ corresponding to the eigenvalues $1$ and $-1$, respectively.

On the other hand, an LHV model can be constructed in order to confirm that the given measurement settings will never allow a violation of any Bell inequality. However, in the two input and two output Bell scenario, there exists an important result establishing an equivalence between the existence of LHV model with the condition that a set of CHSH inequalities always hold~\cite{Fine82,bell-review}. These set of CHSH inequalities contain the inequality \eqref{sipahi} and three more inequalities formed by interchanging, (a) $A_0$ with $A_1$, (b) $B_0$ with $B_1$, and (c) $A_0$ with $A_1$ and $B_0$ with $B_1$ together, in the inequality \eqref{sipahi}. That is, if all these CHSH inequalities always hold, then there must exist an LHV model \cite{Fine82} for the measurement settings corresponding to observable $A_0$, $A_1$, $B_0$, and $B_1$.

\section{Relation between incompatibility and Bell nonlocality}
\label{sec3}
In this section, we explore the connection between incompatible measurements at both parties and the emergence of Bell nonlocality in a simplest Bell scenario. 
The scenario involves two spatially separated parties sharing a quantum state, where each party has two measurement settings, each with two possible dichotomic outcomes.
We begin by considering the case where both parties perform only projective measurements at their local systems.
Any incompatibility of projective measurements, on the sides of both the parties, is sufficient to show Bell nonlocality. That is, when any incompatible projective measurements at both parties are performed, there always exist a quantum shared state for which the parties can share to have Bell nonlocal correlations. However, when the parties consider the local measurements beyond the projective measurements, i.e., POVMs, we find a set of incompatible measurements at both the parties that fail to exhibit Bell nonlocal correlations for any shared state. Thus, incompatibility of POVMs at both parties is not always sufficient for the activation of Bell nonlocality. The Clauser-Horne-Shimony-Holt (CHSH) inequality~\cite{Clauser69} is used to detect Bell nonlocality throughout our work.

\subsection{Projective measurements}

\begin{observation}
\label{theo-1}
In the two-input and two-output Bell scenario with measurements restricted to local projective measurements, any pair of incompatible measurements at both parties is sufficient to demonstrate Bell nonlocality.
\end{observation}

\begin{proof}
Let's consider the CHSH Bell operator, $S \coloneqq A_0\otimes (B_0 + B_1) + A_1 \otimes (B_0 - B_1)$, where $A_{0(1)}$ and $B_{0(1)}$ are observables corresponding to measurements associated with Alice and Bob respectively each having eigenvalues $\pm 1$. 
The squared  CHSH operator can be expressed in the following form \cite{Landau87}:
\begin{align*}
S^2 &= 4I + [A_0, A_1] \otimes [B_0, B_1],\\
S^2 &= 4 (I + J).
\end{align*}
where $I$ is the identity operator and $4J \coloneqq [A_0, A_1] \otimes [B_0, B_1]$. Let the highest eigenvalue of $J$ be denoted as $\mu$. Thus the operator $S^2$ has $4(1+\mu)$ as the highest eigenvalue. Also, since S is Hermitian, $S^2$ has to be positive semidefinite, i.e. all it's eigenvalues are non-negative. These facts imply that
\begin{align*}
   ||S^2|| &= 4(1 + \mu),
\end{align*}
where $||O||:=\sup_{|\psi\rangle }\sqrt{\langle\psi|O^{\dagger}O|\psi}\rangle$, with $|\psi\rangle$ being any pure state. Since $4(1+\mu)$ is the highest eigenvalue of $S^2$, this implies that the maximum absolute value of the eigenvalues of $S$ has to be $2\sqrt{1 + \mu}$. Hence, the maximum CHSH value, when $\mu$ is the highest eigenvalue of operator $J$, is given by~\cite{Landau87}
\begin{equation}
    ||S|| = 2\sqrt{1 + \mu}. \nonumber
\end{equation}

Now, let's assume that the measurement settings by both the parties are non-commutative, meaning that $[A_0, A_1] \neq 0$ and $[B_0, B_1] \neq 0$. This implies that $J$ has at least one non-zero eigenvalue. Also notice that the trace of $J$ is zero, hence the highest eigenvalue of $J$ has to be positive, i.e.,  $\mu>0$.
This implies
\begin{equation}
    ||S||>2,
\end{equation}
when measurement settings by both the parties are non-commutative. Since, non-commutativity of measurements is equivalent to incompatibility of measurements when the measurements are projective ones. Thus, we always get Bell inequality violation for any incompatible projective measurement settings by both the parties. This concludes the proof.

\end{proof}
Bell inequality violation always implies that the local measurements performed by Alice and Bob have to be incompatible because whenever any of the parties perform a jointly measurable sets of measurements, then there always exists an LHV model. With this and observation \ref{theo-1}, it is clear that when both parties perform projective measurements, the incompatibility of the measurements applied by both parties is necessary and can always yield Bell-nonlocal correlations in the two-input and two-output Bell scenario.

\subsection{Non-projective POVM measurements}

\begin{observation}
    In the two-input and two-output Bell scenario, incompatibility of POVMs applied by both parties is not sufficient to exhibit Bell nonlocality.
\end{observation}

\begin{proof}
Let Alice's and Bob's measurements are denoted by  $\lbrace M_{a|x} \rbrace$ and $\lbrace N_{b|y}\rbrace$, respectively, where $a,b=\pm1$ denotes the two outcome and $x,y=0,1$ denotes the two measurement settings. Let us consider Alice and Bob to adopt the following measurement strategies: 
\begin{align*}
M_{\pm 1|0} &=  \frac{\mathbbm{1}\pm\lambda \sigma_z}{2},\\
M_{\pm 1|1} &= \frac{\mathbbm{1}\pm\lambda \sigma_x}{2},\\
N_{\pm 1|0} &=  \frac{\mathbbm{1}\pm\lambda \frac{\sigma_z+\sigma_x}{\sqrt{2}}}{2}, \\
N_{\pm 1|1} &=  \frac{\mathbbm{1}\pm\lambda \frac{\sigma_z-\sigma_x}{\sqrt{2}}}{2},
\end{align*}
where $\lambda \in [0,1]$.
Now consider $A_x \coloneqq M_{+1|x}-M_{-1|x}$ and $B_y \coloneqq N_{+1|y}-N_{-1|y}$.
The CHSH value for given measurement setting for an arbitrary state $\rho$ is then given by
\begin{equation}
    \langle S \rangle_{\rho}=\sqrt{2}\lambda^2\langle\sigma_z\otimes \sigma_z+\sigma_x\otimes \sigma_x\rangle_\rho.
\end{equation}
where $S$ denotes the Bell operator corresponding to the CHSH inequality.
The maximum CHSH value is then given for $\rho=|\phi^+\rangle\langle \phi^+|$, which is the eigenvector of the operator $\sigma_z\otimes \sigma_z+\sigma_x\otimes \sigma_x$ with the highest eigenvalue 2, and thus,
\begin{equation}
\max_{\rho} \langle S \rangle_{\rho}= 2\sqrt{2}\lambda^2.
\end{equation}
In the same way, all other three Bell operators corresponding to the CHSH inequalities (obtained by swapping (a) $A_0$ with $A_1$, (b) $B_0$ with $B_1$, and (c) $A_0$ with $A_1$ and $B_0$ with $B_1$ together in the inequality \eqref{sipahi}) will yield the highest eigenvalue as $2\sqrt{2}\lambda^2$.
This implies that if $\lambda \leq \frac{1}{2^{1/4}}$, then none of these CHSH inequalities will get violated, 
and thus for the given measurements with $\lambda \leq \frac{1}{2^{1/4}}$, there must exist a LHV model for any shared state \cite{Fine82}. Moreover, it is known that these measurement settings are incompatible when $\lambda > \frac{1}{\sqrt{2}}$~\cite{incompatibility-review}. This implies that for $\frac{1}{\sqrt{2}}<\lambda\leq\frac{1}{2^{1/4}}$, even though the POVMs on both parties are incompatible, they can never activate Bell nonlocality in the two-input and two-output Bell scenario. This concludes the proof. 
\end{proof}

\section{Sufficient condition for Bell nonlocality for fixed amount of measurement incompatibility and entanglement of pure two qubits}
\label{sec4}
In this section, our aim is to find a sufficient condition for Bell nonlocality in the two input and two output Bell scenario for a fixed value of entanglement and given the measurements have a fixed degree of incompatibility. For simplicity, we restrict the study to the case where the shared state is pure and acts on the Hilbert space of dimension $2 \otimes 2$. We also assume that the measurements performed are projective. We find a sufficient condition by obtaining the maximal CHSH value given the fixed amount of entanglement and measurement incompatibility.

Any two-qubit pure state, $|\psi\rangle$, up to a local unitary can be expressed as follows:
\begin{align}\label{saral}
|\psi\rangle = \sqrt{E} \ket{00} + \sqrt{1-E} \ket{11}, 
\end{align}
where $E$ belongs to the interval $[0, \frac{1}{2}]$, and $|0\rangle$ and $|1\rangle$ are orthogonal single-qubit pure states.  The state, $|\psi\rangle$ is product for $E=0$ and maximally entangled when it is equal to $\frac{1}{2}$. Also, the parameter $E$ is an entanglement monotone~\cite{nielsen-prl,vidal-prl,jonathan-plenio-first,hardy-area,vidal-jmo,lo-popescu}, and thus is taken as a measure of entanglement of pure two-qubit states. Now the observables representing the projective measurement settings, acting on a Hilbert space of dimension 2 and with  the outcomes $\pm 1$, applied by Alice and Bob can be expressed as
\begin{align}
    A_r&=\hat{a}_r.\vec{\sigma}; ~r=0,1,\nonumber\\
    B_s&=\hat{b}_s.\vec{\sigma}; ~s=0,1,
\end{align}
respectively where $\vec{\sigma}=(\sigma_x,~\sigma_y,~\sigma_z)$ and $\hat{a}_r(\hat{b}_s)$ represent unit vectors in the three-dimensional Euclidean space $\mathbbm{R}^3$. The CHSH operator in terms of these observables is given by,
\begin{align}
    S=(\hat{a}_0 \cdot \vec{\sigma} \otimes (\hat{b}_0 +\hat{b}_1) \cdot \vec{\sigma} 
+ \hat{a}_1 \cdot \vec{\sigma} \otimes (\hat{b}_0-\hat{b}_1) \cdot \vec{\sigma}). \nonumber
\end{align}
In the following subsections, we consider two cases. In the first case, we consider a fixed amount of incompatibility of the measurements applied by both parties. For the second case, we fix the amount of incompatibility of the measurements applied by a single party only. We refer to these two cases as ``both sided two-input and two-output Bell scenario" and ``one sided two-input and two-output Bell scenario" respectively.

\subsection{Both sided two-input and two-output Bell scenario}
The degree of non-commutativity of the measurement settings applied by the parties in the both sided two-input and two-output Bell scenario can be defined as
\begin{align}
    \Delta& \coloneqq \frac{1}{4}||[A_0,A_1][B_0,B_1]||.
\end{align}
Now we are set to state the following:
\begin{observation}
\label{kaun-hu-main}
In a both sided two-input and two-output Bell scenario, assuming the parties share pure two-qubit states with fixed entanglement $E$ and perform local incompatible projective measurements,
the criterion for establishing Bell nonlocal correlations 
is given by $(2-X) \sqrt{1+ \Delta} + X \sqrt{1-\Delta} > 2$ where $X \coloneqq 1 - 2\sqrt{E(1-E)}$, 
and $\Delta$ quantifies the degree of incompatibility of measurements performed by both parties.
\end{observation}

\begin{proof}
The objective is to compute the maximum CHSH value over the pure two-qubit states with a fixed amount of entanglement, $E$, given a fixed amount of incompatibility of measurements performed by both parties. Any pure two-qubit state with $E$ amount of entanglement can be written as $U_1 \otimes U_2 |\psi\rangle$, where $|\psi\rangle$ is given in Eq. \eqref{saral}, has also $E$ amount of entanglement.  
$U_1$ and $U_2$ represent unitaries acting on two parts of $|\psi\rangle$ locally.
Hence, we look for the CHSH value, maximized over local unitaries on $|\psi\rangle$, that is,
\begin{align}
\label{eqn-partial-2}
F_1 &\coloneqq \max_{U_1,U_2} \Tr{[U_1 \otimes U_2 |\psi\rangle \langle \psi| U_1^\dagger \otimes U_2^\dagger S ]} \nonumber \\
&= \max_{U_1,U_2} \Tr[ U_1 \otimes U_2 |\psi\rangle \langle \psi| U_1^\dagger \otimes U_2^\dagger (\hat{a}_0 \cdot \vec{\sigma} \otimes (\hat{b}_0 +\hat{b}_1) \cdot \vec{\sigma} \nonumber \\
&+ \hat{a}_1 \cdot \vec{\sigma} \otimes (\hat{b}_0-\hat{b}_1) \cdot \vec{\sigma})].
\end{align}
Since Alice and Bob have two measurement settings (projective) each, without loss of any generality, the unit vectors of the observables corresponding to these measurements can be taken as,
$\hat{a}_0 = \hat{z}$, $\hat{a}_1 = \hat{z} \cos \phi + \hat{x} \sin \phi$, $\hat{b}_0 = \hat{z}\cos \frac{\theta}{2}+\hat{x}\sin\frac{\theta}{2} $, and $\hat{b}_1 =  \hat{z}\cos \frac{\theta}{2}-\hat{x}\sin\frac{\theta}{2}$, where $\theta,\phi \in [0,\frac{\pi}{2}]$. 
$\hat{z}, \hat{x}$ denote two unit vectors $(0,0,1)$ and $(1,0,0)$ in the three dimensional Cartesian coordinates respectively.
Then, the degree of incompatibility of measurements applied by both parties is given by, $\Delta=\sin \theta \sin \phi$.
We take the parametric form of unitaries, namely $U_1$ and $U_2$, as detailed in Appendix~\ref{appendix}. After substituting the expressions for $U_1$ and $U_2$ into Eq.~\eqref{eqn-partial-2}, the function $F \coloneqq \Tr{[U_1 \otimes U_2 |\psi\rangle \langle \psi| U_1^\dagger \otimes U_2^\dagger S ]}$ becomes parameterized by the variables: $\theta, \phi, \psi_1, \phi_1, \theta_1, \psi_2, \theta_2,$ and $\phi_2$ where $\psi_1, \phi_1, \theta_1, \psi_2, \theta_2,$ and $\phi_2$ are the parameters of local unitaries, $U_1$ and $U_2$.

In order to find the global maximum of $F$ with respect to the parameters of local unitaries, we look for the solutions of the equations where partial derivatives of $F$ is zero with respect to the parameters of unitaries, and then put back the one of these solutions in the expression of $F$ which gives it's maximal value. Firstly, we set $\frac{\partial F}{\partial \phi_1 } = 0$ and plug equation obtained from it into the expression for $F_1$, we ascertain that the maximum of $F$ over $\psi_1$ and $\psi_2$ is achieved when $\psi_1 + \psi_2$ is either $0$ or $\pi$. Opting for the case $\psi_1 + \psi_2 = 0$, we introduce a modified function $F_2 \coloneqq \max_{\psi_1,\psi_2} F$. The next step involves maximizing $F_2$ over $\phi_1,\theta_1,\phi_2,$ and $\theta_2$ to yield $F_1$. Continuing the analysis, upon further substitution of obtained equations from $\frac{\partial F_2}{\partial \phi_1 } = 0$ into $\frac{\partial F_2}{\partial \phi_2} = 0$, we conclude that $\phi_1 = 0$. It is noteworthy that while alternative solutions might exist, they will not yield the global maximum of $F_2$. Consequently, enforcing $\phi_1 = 0$ implies $\phi_2 = 0$ via the deduced $\frac{\partial F_2}{\partial \phi_1 } = 0$.
Continuing onward, substituting $\phi_1 = \phi_2 = 0$ into the equations $\frac{\partial F_2}{\partial \theta_1 } + \frac{\partial F_2}{\partial \theta_2} = 0$ and $\frac{\partial F_2}{\partial \theta_1 } - \frac{\partial F_2}{\partial \theta_2} = 0$, leads to the following relationships:
\begin{align}
\sin(\theta_1 + \theta_2) &= \frac{-\sin \frac{\theta}{2}\cos \phi}{\sqrt{1-\sin \theta \sin \phi}},\\
\cos (\theta_1 + \theta_2) &= \frac{\cos \frac{\theta}{2}- \sin \frac{\theta}{2} \sin \phi}{\sqrt{1-\sin \theta \sin \phi}},\\
\sin (\theta_1 - \theta_2) &= \frac{\sin \frac{\theta}{2}\cos \phi}{\sqrt{1+\sin \theta \sin \phi}},\\
\cos (\theta_1 - \theta_2) &= \frac{\cos \frac{\theta}{2}+ \sin \frac{\theta}{2} \sin \phi}{\sqrt{1+\sin \theta \sin \phi}}.
\end{align}

Upon substituting these equations, along with $\phi_1 = \phi_2 = 0$, into $F_2$, we obtain

\begin{align}
F_1 &= (2-X) \sqrt{1+\sin \theta \sin \phi} + X \sqrt{1-\sin \theta \sin \phi},\nonumber \label{final-eqn}\\
&= (2-X) \sqrt{1+\Delta} + X\sqrt{1-\Delta}.
\end{align}
Here, $X = 1 - 2\sqrt{E(1-E)}$. This formula yields the maximum values of $F$ by considering all feasible unitaries $U_1$ and $U_2$. Furthermore, it can be observed that $\psi_1 + \psi_2 = \pi$ does not result in the global maximum of $F$.
It concludes the proof.
\end{proof}

\begin{corollary}
In a both sided two-input and two-output Bell scenario, any degree of incompatibility in the projective measurements performed by both parties is sufficient to violate Bell nonlocality, provided the entanglement of the shared state satisfies $E > \frac{1}{2}(1-\sqrt{2\sqrt{2}-2})$. 
\end{corollary}

\begin{corollary}
The maximum CHSH value exhibits a non-monotonic relationship with the degree of incompatibility of the projective measurements performed by both parties in a both sided two-input and two-output Bell scenario, when they share a partially entangled pure state with entanglement $0< E < \frac{1}{2}$. However, when $E=\frac{1}{2}$, the maximum CHSH value varies monotonically with the degree of incompatibility of measurements applied by both parties.
\end{corollary}

\subsection{One-sided two-input and two-output Bell scenario}
Here, we focus on a two-input and two-output Bell scenario involving projective measurements at both sides, each acting on a two-dimensional Hilbert space, and the state shared among the parties is a pure two-qubit state. 
Our goal is to establish a sufficient condition for activating Bell nonlocality, given a fixed amount of shared entanglement and a fixed amount of measurement incompatibility on one side, by optimizing over all possible measurement settings on the other side.

\begin{observation}
In the two-input and two-output Bell scenario, where the parties can perform only projective measurements and the allowed shared state is a pure two-qubit state, given a fixed amount of shared entanglement $E$ and a fixed amount of incompatibility in the measurement settings of one party $\lambda$, the criterion for violating the Bell inequality with at least one pair of measurement settings on the other party is as follows:
For $X \leq X_c \coloneqq 2 \frac{\sqrt{1 - \lambda}}{\sqrt{1 + \lambda} + \sqrt{1 - \lambda}}$, the criterion is:
\begin{align}
    (2 - X) \sqrt{1 + \lambda} + X \sqrt{1 - \lambda} > 2
\end{align}
Otherwise, the condition is:
\begin{align}
    \frac{2X}{\sqrt{1 - \lambda \sin{\phi^\ast}}} > 2
\end{align}
where
$X \coloneqq 1 - 2\sqrt{E(1 - E)}$ and $\phi^\ast$ is the solution to the equation, $\frac{X}{\sqrt{1 - \lambda \sin{\phi^\ast}}} = \frac{2 - X}{\sqrt{1 + \lambda \sin{\phi^\ast}}}$.
\end{observation}

\begin{proof}
Here, we address the same concept as in Observation~\ref{kaun-hu-main}, with the key difference being that the fixed amount of measurement incompatibility is considered only for one party (Alice). Consequently, we maximize the CHSH value over all local unitaries acting on fixed pure two-qubit entanglement and all incompatible projective measurements of the other party (Bob). The degree of incompatibility of measurements applied by Alice is defined as $\lambda \coloneqq \frac{1}{2}||[A_0,A_1]||$.
Similar to Observation~\ref{kaun-hu-main}, without any loss of generality, the unit vectors corresponding to the observables of the measurements can be chosen as:  
$
\hat{a}_0 = \hat{z}, \quad \hat{a}_1 = \hat{z} \cos \phi + \hat{x} \sin \phi, \quad \hat{b}_0 = \hat{z} \cos \frac{\theta}{2} + \hat{x} \sin \frac{\theta}{2}, \quad \text{and} \quad \hat{b}_1 = \hat{z} \cos \frac{\theta}{2} - \hat{x} \sin \frac{\theta}{2},
$ 
where $\theta, \phi \in [0, \frac{\pi}{2}]$. From these definitions, the degree of incompatibility of measurement settings employed by Alice is given by $\lambda = \sin \theta$, and  $\Delta$ defined in Observation~\ref{kaun-hu-main} will be
$\Delta = \lambda \sin \phi$.
To derive a sufficient criterion for achieving Bell nonlocal correlations given a fixed amount of pure two-qubit entanglement, $E$, and a fixed incompatibility of measurements employed by Alice, $\lambda$, we just need to optimize $F_1$ with respect to $\phi$ since optimization of CHSH value over the parameters of local unitaries acting on the shared pure two-qubit state with entanglement $E$ has already been addressed in Observation~\ref{kaun-hu-main}.
To achieve this, we first find the conditions for $\frac{\partial F_1}{\partial \phi} = 0$. We have obtained two such conditions, and substituting them into $F_1$ gives the expressions for $T_1$ and $T_2$ as follows:
$T_1 \coloneqq (2 - X) \sqrt{1 + \lambda} + X \sqrt{1 - \lambda}$
and $T_2 \coloneqq \frac{2X}{\sqrt{1 - \lambda \sin{\phi^\ast}}}$,
where $X \coloneqq 1 - 2\sqrt{E(1 - E)}$ and $\phi^\ast$ is the solution to the equation
$
\frac{X}{\sqrt{1 - \lambda \sin{\phi^\ast}}} = \frac{2 - X}{\sqrt{1 + \lambda \sin{\phi^\ast}}}$.

Next, we observe the following properties:
\begin{itemize}
    \item For all $\lambda \neq 0$, $T_1$ is a monotonically decreasing function of $X$, while $T_2$ increases monotonically with $X$.
    \item At $X = 0$, $T_1 > T_2$, and at $X = 1$, $T_2 > T_1$ for all $\lambda \neq 0$.
    \item $T_1$ and $T_2$ intersect at
    $X = 2 \frac{\sqrt{1 - \lambda}}{\sqrt{1 + \lambda} + \sqrt{1 - \lambda}}$.
\end{itemize}

This completes the proof.
\end{proof}

\begin{corollary}
For each set of incompatible projective measurements performed by one party, there exists at least one pair of measurement settings on the other party that leads to a Bell inequality violation if and only if the shared state is maximally entangled, i.e., $E= \frac{1}{2}$.
\end{corollary}

\section{Conclusion}
\label{sec5}
Bell nonlocality is a non-classical phenomenon that cannot be explained by any local hidden variable theory.
Incompatibility of measurements at both parties are necessary for the activation of Bell nonlocality in arbitrary Bell scenarios, regardless of the shared state. Moreover, in the two-input and two-output Bell scenario, it was known that any pair of incompatible measurements at one party ensures the existence of at least one shared state and one pair of measurements at the other party, for which violation of a Bell inequality can be attained~\cite{incomp-bell}. However, 
in Bell scenarios involving more than two measurement settings per party, there exist works exhibiting that measurement incompatibility, employed by a single party, does not always guarantee Bell nonlocality~\cite{imcomp-bell-not-1,incomp-bell-not-2}.

In this paper, we ask a more specific question in the two-input and two-output Bell scenario: Can any pair of measurement incompatibility employed by each party ensure Bell nonlocal correlations for at least one shared state?  
We find that there exists a set of incompatible measurements for each party such that, if both parties apply these measurements on any shared state, Bell nonlocality cannot be activated in this setup. Obtaining this negation required the use of positive operator-valued measurements at both parties. 
However, when restricted to projective measurements, the question has an affirmative answer.

Next, we investigate the interplay between a fixed amount of shared entanglement and a fixed amount of measurement incompatibility required to activate Bell nonlocality in the two-input and two-output Bell scenario.
Our study includes determining a sufficient criterion for Bell inequality violation under the assumption that both the shared entanglement and the measurement incompatibility at both parties are fixed. Additionally, given a fixed amount of shared entanglement and a fixed amount of incompatibility in the measurements at one party, we explore a condition under which the Bell inequality violation will occur with at least one pair of measurements at the other party. These analyses allow us to determine the minimum shared entanglement required to activate Bell nonlocality for any incompatible measurements applied by either single party or both parties in this Bell scenario.

\section*{Acknowledgment}
We thank Edwin Peter Lobo for useful discussions.
CS acknowledges funding from the QuantERA II Programme that has received funding from the European Union’s Horizon 2020 research and innovation programme under Grant Agreement No. 101017733 and the F.R.S.-FNRS Pint-Multi programme under Grant Agreement R.8014.21.
from the European Union’s Horizon Europe research and innovation programme under the project "Quantum Security Networks Partnership" (QSNP, grant agreement No 101114043), from the F.R.S-FNRS through the PDR T.0171.22, from the FWO and F.R.S.-FNRS under the Excellence of Science (EOS) programme project 40007526, from the FWO through the BeQuNet SBO project S008323N. Funded by the European Union. Views and opinions expressed are, however, those of the authors only and do not necessarily reflect those of the European Union, who cannot be held responsible for them.
We acknowledge partial support from the Department of Science and Technology, Government of India through the QuEST grant (grant number DST/ICPS/QUST/Theme3/2019/120).

\section{appendix}
\label{appendix}
The form of general unitaries $U_1$ and $U_2$ in $2\otimes 2$ system are given as follows:
\begin{widetext}

\begin{equation*}
U_1 =  \left( \begin{array}{cc}
\cos \frac{\theta_1}{2}\exp{
\frac{i}{2}(\psi_1+\phi_1)} & \sin \frac{\theta_1}{2} \exp{-\frac{i}{2}(\psi_1-\phi_1)} \\
-\sin \frac{\theta_1}{2} \exp{\frac{i}{2}(\psi_1-\phi_1)} & \cos \frac{\theta_1}{2}\exp{-
\frac{i}{2}(\psi_1+\phi_1)} \end{array} \right),    
\end{equation*}

 \begin{equation*}
U_2 =  \left( \begin{array}{cc}
\cos \frac{\theta_2}{2}\exp{
\frac{i}{2}(\psi_2+\phi_2)} & \sin \frac{\theta_2}{2} \exp{-\frac{i}{2}(\psi_2-\phi_2)} \\
-\sin \frac{\theta_2}{2} \exp{\frac{i}{2}(\psi_2-\phi_2)} & \cos \frac{\theta_2}{2}\exp{-
\frac{i}{2}(\psi_2+\phi_2)} \end{array} \right),    
\end{equation*}

\end{widetext}
where $\psi_1,\psi_2 \in [0,4\pi)$, $\phi_1,\phi_2 \in [0,2\pi]$, and $\theta_1,\theta_2 \in [0,\pi]$. 

\bibliography{Bell}
\end{document}